\documentclass{article}
\usepackage{graphicx} % Required for inserting images
\usepackage{amsthm, amsmath, amsfonts,mathtools, fullpage, graphicx, wrapfig, tikz, caption, subcaption}
\usepackage{algorithm}
\usepackage[noend]{algpseudocode} % nice, compact pseudo
\algrenewcommand\algorithmicrequire{\textbf{Input:}}
\algrenewcommand\algorithmicensure{\textbf{Output:}}

\usepackage{natbib}
\usepackage{amssymb}
\usepackage{verbatim}
\usepackage{relsize}
\usepackage{listings}
\usepackage{bbm}
\usepackage{xcolor}
\usepackage{hyperref}
\usepackage{graphicx}
\usepackage{multirow}
\usepackage{authblk}
\usepackage{booktabs}
\usepackage{siunitx}

\newtheorem{theorem}{Theorem}
\newtheorem{lemma}{Lemma}

\newcommand{\E}{\mathbb{E}}

\usepackage{setspace}

 \usepackage[vmargin=1.18in,hmargin=1.1in]{geometry}

\title{Structural Nested Mean Models for Modified Treatment Policies}
\author[1]{Zach Shahn}

\affil[1]{CUNY School of Public Health, New York, NY, USA}
\begin{document}

\maketitle

\begin{abstract}
There is a growing literature on estimating effects of treatment strategies based on the natural treatment that would have been received in the absence of intervention, often dubbed `modified treatment policies' (MTPs). MTPs are sometimes of interest because they are more realistic than interventions setting exposure to an ideal level for all members of a population. In the general time-varying setting, \citet{richardson2013single} provided exchangeability conditions for nonparametric identification of MTP effects that could be deduced from Single World Intervention Graphs (SWIGs). \citet{diaz2023nonparametric} provided multiply robust estimators under these identification assumptions that allow for machine learning nuisance regressions. In this paper, we fill a remaining gap by extending Structural Nested Mean Models (SNMMs) \citep{jamie1994,robins2004optimal,snmm_review2014} to MTP settings, which enables characterization of (time-varying) heterogeneity of MTP effects. We do this both under the exchangeability assumptions of \citet{richardson2013single} and under parallel trends assumptions, which enables investigation of (time-varying heterogeneous) MTP effects in the presence of some unobserved confounding. 
\end{abstract}

\section{Introduction}
There is a growing literature on estimating effects of treatment strategies based on the natural treatment that would have been received in the absence of intervention \citep{robins2004effects,richardson2013single,young2014identification,haneuse2013estimation,munoz2012population,diaz2023nonparametric,sani2020identification}. Examples of such strategies include `exercise 20 minutes longer than you normally would' or `discharge patients from the intensive care unit one day later than under usual care'. Strategies depending on the natural value of treatment have been dubbed `modified treatment policies' (MTPs), a term which has gained some traction and we hence adopt here. MTPs are sometimes of interest because they are more realistic than interventions setting exposure to an ideal level for all members of a population. 

In the general time-varying setting, \citet{richardson2013single} provided exchangeability conditions for nonparametric identification of MTP effects that could be deduced from Single World Intervention Graphs (SWIGs). \citet{diaz2023nonparametric} provided multiply robust estimators under these identification assumptions that allow for machine learning nuisance regressions. In this paper, we fill a remaining gap by extending Structural Nested Mean Models (SNMMs) \citep{jamie1994,robins2004optimal,snmm_review2014} to MTP settings, which enables characterization of (time-varying) heterogeneity of MTP effects. We do this both under the exchangeability assumptions of \citet{richardson2013single} and under parallel trends assumptions \citep{shahn2022structural}, which enables investigation of (time-varying heterogeneous) MTP effects in the presence of some unobserved confounding. 

The structure of the paper is as follows. In Section \ref{section_notation}, we summarize notation and review definitions of MTPs and SNMMs. Once defined, we also provide some hypothetical examples of when SNMMs for MTP effects might be of interest. In Section \ref{section_exchangeability}, we provide identification results and Neyman orthogonal estimators under exchangeability assumptions. In Section \ref{section_pt}, we do the same under parallel trends assumptions. In Section \ref{sec:sim}, we provide simulations showing that our estimators are unbiased and that the sandwich variance estimates provide nominal coverage. In Section \ref{section:application}, we present a real data analysis estimating the effects of shifting mobility from its natural level on subsequent county level Covid-19 incidence. In Section \ref{sec:disc}, we conclude and discuss some directions for future work. 

\section{Notation, MTPs, and SNMMs}\label{section_notation}
Suppose we observe a cohort of $N$ subjects indexed by $i \in \{1,\ldots,N\}$. Assume that each subject is observed at regular intervals from baseline
time $0$ through end of follow-up time $K$, and there is no loss to
follow-up. At each time point $t$, the data are collected on $O
_t=(Z_t,Y_t,A_t)$ in that temporal order. $A_t$ denotes the (possibly
multidimensional with discrete and/or continuous components) treatment
received at time $t$, $Y_t$ denotes the outcome of interest at time $t$, and 
$Z_t$ denotes a vector of covariates at time $t$ excluding $Y_t$. Hence, $Z_0$
constitutes the vector of baseline covariates other than $Y_0$. For
arbitrary time varying variable $X$: we denote by $\bar{X}_t=(X_0,\ldots,X_t)
$ the history of $X$ through time $t$; we denote by $\underline{X}%
_t=(X_t,\ldots,X_K)$ the future of $X$ from time $t$ through time $K$; and
whenever the negative index $X_{-1}$ appears it denotes the null value with
probability 1. In Section \ref{section_exchangeability}, when we work under exchangeability assumptions, we define $\bar{L}_t$ to be $(\bar{Z}_t,\bar{Y}_{t})$, i.e. the joint covariate and outcome history through time $t$, but we do not require that the outcome is necessarily measured at all time points. In Section \ref{section_pt}, when we work under parallel trends assumptions, we define $\bar{L}_t$ to be $(\bar{Z}_t,\bar{Y}_{t-1})$, i.e. the joint covariate and outcome history through time $t$ excluding the most recent outcome $Y_t$, and we assume that the outcome is measured at each time point. Hence, in the parallel trends section, $L_0$ is 
$Z_0$. In all sections, we let $H_t=(\bar{L}_t,\bar{A}_{t-1})$ denote the relevant pre-treatment history. We denote random variables by capital letters and their realizations using lower case letters. We adopt the counterfactual framework for time-varying treatments \citep{robins1986new} which posits that corresponding to each time-varying
treatment regime $\bar{a}_t$, each subject has a counterfactual or potential
outcome $Y_{t+1}(\bar{a}_{t})$ that would have been observed had that
subject received treatment regime $\bar{a}_t$.

Let $g=(g_1,\ldots,g_K)$ denote an arbitrary MTP with each $g_t: (h_t,a_t)\rightarrow a_t^{+}$ a treatment rule setting the modified treatment ($a_t^{+}$) as a function of observed history through $t$. $A_t(g)$ is the treatment that would naturally occur at time $t$ had strategy $g$ been imposed through $t-1$. $A_t^{+}(g)$ is the treatment received at time $t$ under $g$, possibly a function of $A_t(g)$. $Y_t(g)\equiv Y_t(A_{t-1}^{+}(g))$ is the counterfactual outcome had treatment been assigned according to $g$. We will use the shorthand $Y_t(\bar{A}_m,\underline{g})$ to denote the counterfactual outcome under the observed regime through time $m$ followed by regime $g$ thereafter. 

A general regime SNMM for MTP $g$ models the contrasts
\begin{equation}\label{snmm}
    \gamma_{tk}^g(h_t,a_t)=E[Y_k(\bar{A}_{t},\underline{g})-Y_k(\bar{A}_{t-1},\underline{g})|H_t=h_t,A_t=a_t].
\end{equation}
(\ref{snmm}) represents the conditional lasting effects among subjects with observed history $(h_t,a_t)$ of receiving the observed treatment at time $t$ then switching to MTP $g$ thereafter, compared to switching to $g$ at time $t$ and continuing to follow it thereafter. Thus, (\ref{snmm}) might be interpreted as the conditional effect of one final blip of the observational regime before switching to an MTP. 

With knowledge of $\gamma_{tk}^g$, we can `blip down' or strip away these effects from observed outcomes to obtain consistent estimates of conditional counterfactual outcomes under the MTP $g$.

\begin{lemma}\label{lemma1}
  $E[Y_k-\sum_{j=t}^{k-1} \gamma_{j}(H_j,A_j)|H_t,A_t]=E[Y_k(\bar{A}_{t-1},\underline{g})|H_t,A_t]$
\end{lemma}
The proof of Lemma \ref{lemma1} is in the Appendix. It follows as a special case of Lemma \ref{lemma1} that given $\gamma_{tk}^g$, $E[Y_k(g)]$ is identified as $E[Y_k-\sum_{j=0}^{k-1}\gamma_{jk}(A_j,H_j)]$. 

A parametric SNMM specifies a functional form
\begin{equation}\label{snmm_par}
\gamma_{tk}^g(h_t,a_t)=\gamma_{tk}^g(h_t,a_t;\psi^*)
\end{equation}
with $\gamma_{tk}^g(h_t,a_t;\psi)$ a known function of finite dimensional parameter $\psi$ taking unknown true value $\psi^*$ such that $\gamma_{tk}^g(h_t,a_t;\psi)=0$ whenever $g(h_t,a_t)=a_t$ or $\psi=0$. 

%As one substantive example, suppose $A_t$ represents average caloric intake in year $t$, and the regime $g(h_t,a_t)=\mathbf{1}\{\lambda a_t>\delta\}\lambda a_t + \mathbf{1}\{\lambda a_t\leq\delta\}a_t$ represents a diet that reduces caloric intake by a multiplicative factor $\lambda$ relative to its natural level as long as doing so does not reduce intake below threshold $\delta$. Then (\ref{snmm}) represents the conditional lasting effects of one extra year (year $t$) before beginning the diet $g$. The parametric model (\ref{snmm_par}) might, for example, capture that the effects of delayed dieting tend to be more harmful at higher observed caloric intake and lower.  

As one substantive example, suppose $A_t$ is a continuous measure of medication adherence (e.g. proportion of prescribed doses taken), and it is thought that efficacy starts significantly declining below some threshold $\delta$. One might be interested in the effect of a partial adherence enforcing intervention $g(h_t,a_t)=\mathbf{1}\{a_t<\delta\}\delta+\mathbf{1}\{a_t\geq\delta\}a_t$ that sets treatment to $\delta$ if its natural value is below that threshold and leaves the natural value unchanged otherwise. In this setting, $\gamma_{tk}(h_t,a_t)$ is the lasting effect of adherence level $a_t$ versus $g(h_t,a_t)$ in month $t$ followed by regime $g$ thereafter. If $a_t\geq\delta$, then the effect would be $0$. For $a_t<\delta$, parametric model (\ref{snmm_par}) could characterize how the lasting effects of a blip of non-adherence relative to regime $g$ depend on the magnitude of nonadherence ($\delta-a_t$) and health history. 

As another example, this time concerning an exposure more often studied under parallel trends assumptions, suppose $A_t$ denotes state minimum wage in year $t$. Perhaps interest centers on the effect on poverty of increasing minimum wage by $\$2$ above its natural value if its natural value is less than $\$10$. (\ref{snmm}) then defines the conditional effects on poverty of natural wages less than $\$10$ relative to their value had they been increased by $\$2$. (\ref{snmm}) might, for example, reveal that $\$2$ wage hikes relative to observed wages would have been more impactful in states that actually had lower wages. This is a different estimand than the effect of treatment on the treated (i.e. the effect of wage hikes from the previous year in states that implemented them) that economists typically study with Difference in Differences.

\section{Identification and Estimation Under Exchangeability Assumptions}\label{section_exchangeability}
 \subsection{Point Exposure Setting}
In the point exposure setting where $K=1$, we will drop time subscripts for simplicity. We will make the standard consistency assumption that 
\begin{equation}\label{pe_consistency}
    \textbf{Consistency: } Y(a)=Y \text{ whenever }A=a.
\end{equation}
We also make a positivity assumption with respect to regime $g$
\begin{equation}\label{pe_positivity}
    f_{A,L}(a,l)>0 \text{ implies } f_{A,L}(g(a,l),l)>0,
\end{equation}
which states that for any treatment and covariate values that might be observed in the data, the corresponding $g$-modified treatment value with the same covariates might also be observed. Finally,
\citet{haneuse2013estimation} introduced the MTP exchangeability assumption
\begin{equation}\label{pe_seqex}
E[Y(a')|A=a,L=l]=E[Y(a')|A=a',L=l] \text{ for all } a'=g(a,l)
\end{equation}
and showed that under assumptions (\ref{pe_consistency}), (\ref{pe_positivity}), and (\ref{pe_seqex}), conditional MTP effects are identified by 
\begin{equation}
E[Y(g)|A=a,L=l]=E[Y|A=g(a,l),L=l].
\end{equation}
Identification of a SMM follows immediately:
\begin{align}\label{pe_ex_id}
    \gamma^g(l,a)\equiv E[Y-Y(g(l,a))|A=a,L=l]=E[Y|A=a,L=l]-E[Y|A=g(a,l),L=l].
\end{align}
Consider a parametric SMM $\gamma^g(l,a;\psi^*)$ as defined in (\ref{snmm_par}). Then (\ref{pe_ex_id}) implies that $\psi^*$ satisfies
\begin{equation}
    E[q(L,A)(Y-\gamma^g(L,A;\psi^*)-\mu(g(A,L),L))]=0
\end{equation}
where $\mu(a,l)=E[Y|A=a,L=l]$. Thus, a consistent and asymptotically normal estimator of $\psi^*$ can be obtained by solving the corresponding estimating equations 
\begin{equation}\label{eq:outcome_regression_ex} \mathbb{P}_N\begin{bmatrix}b(A,L)\{Y-\mu(A,L;\beta)\}\\
q(A,L)\{Y-\gamma^g(A,L;\psi)-\mu(g(A,L);\beta)\}\end{bmatrix}=0
\end{equation}
where $\mu(A,L;\beta)$ is a model for $\mu(A,L)$ smooth in an $r$-dimensional parameter $\beta$ (equal to $\beta^*$ under the true law) and $q(A,L)$ and $b(A,L)$ are conformable analyst selected index functions. 

The estimator solving (\ref{eq:outcome_regression_ex}) is very sensitive to misspecification of $\mu$, and asymptotic normality may not hold if, to avoid misspecification, $\mu$ is estimated by machine learning. Therefore, we introduce an estimator that has the Neyman-orthogonality property \citep{chernozhukov2018double} crucial for enabling machine learning estimation of nuisance functions.

Let $T_g$ act on functions $h$ by $(T_gh)(A,L)=h(g(A,L),L)$, and define the $L$-conditional inner product weighted by the observed treatment density:
\[
\langle u,v\rangle \;=\; E\!\left[\int u(L,a)\,v(L,a)\,f_{A\mid L}(a\mid L)\,da\right].
\]
Let $T_g^{\dagger}$ be the adjoint to $T_g$ such that
\[
\langle q, T_g h\rangle_w = \langle T_g^{\dagger} q, h\rangle_w.
\]
For a shift $g(a)=a+\delta$,
\[
(T_g)^{\dagger} q(L,A)\;=\; q\!\big(L,g(A)\big)\,\frac{f_{A\mid L}(A-\delta\mid L)}{f_{A\mid L}(A\mid L)}.
\]
In particular, for $\tilde q := T_g^{\dagger}q$ we have the identity
\begin{equation}
\label{eq:adj_identity}
E[q(A,L)\,h\!\big(g(A,L),L\big)\,f_{A|L}(A|L)]
=
E\!\left[\tilde q(A,L)\,h(A,L)\,f_{A|L}(A|L)\right] \quad \forall h.
\end{equation}
Thus, $T_g^{\dagger}$ is essentially a change of variables operator for $g$.
 If $g(\cdot,l)$ is bijective and differentiable, then
\[
\tilde q(l,a)= q\!\big(l,g^{-1}(a,l)\big)\,
\frac{\pi\!\big(g^{-1}(a,l)\mid l\big)}{\pi(a\mid l)}\,
\big|\det J_a g^{-1}(a,l)\big|.
\]
If $g(\cdot,l)$ is possibly many-to-one and $A$ is discrete,
\[
\tilde q(l,a)=\sum_{a':\,g(a',l)=a} q(l,a')\,\frac{\pi(a'\mid l)}{\pi(a\mid l)}.
\]
\begin{theorem}(Neyman orthogonal estimator, point exposure)\label{theorem1}
\label{thm:point_orthogonality_full}
Assume Consistency, Positivity, MTP Exchangeability, and regularity conditions.
Consider the score
\begin{equation}
\label{eq:point_score}
\phi(O;\psi,\eta)
=\underbrace{(q-\tilde q)\,(Y-\mu(A,L))}_{\text{augmentation}}
+\underbrace{q\{\mu(A,L)-\mu(g(A,L),L)-\gamma_g(L,A;\psi)\}}_{\text{identifying}},
\quad \eta=(\mu,\pi).
\end{equation}
Then:
\begin{enumerate}
\item \textbf{Identification.} At $(\psi^\ast,\eta^\ast)$,
\(
E[\phi(O;\psi^\ast,\eta^\ast)]=0.
\)
\item \textbf{Neyman orthogonality.}
For any regular parametric submodel $t\mapsto \eta_t=(\mu_t,\pi_t)$ with $\eta_0=\eta^\ast$,
\[
\left.\frac{d}{dt}\right|_{t=0} E\!\left[\phi(O;\psi^\ast,\eta_t)\right]=0.
\]
\item \textbf{Asymptotic variance} Assume that $\hat\eta$ is learned on held-out folds and $||\hat\eta-\eta^{*}||_{L_2}=o_p(1)$.
 If $\dim(q)=\dim(\psi)=d$ and $\hat\psi$ solves
$\frac{1}{n}\sum_{i=1}^n \hat\phi_i(\psi)=0$, then for
\[
G:=E[\frac{\partial}{\partial\psi^\top}\phi(O;\psi^{*},\eta^{*})]\in R^{d\times d} 
\]
and
\[
\Sigma:=Var(\phi(O;\psi^{*},\eta^{*})\big)\in R^{d\times d},
\]
\[
\sqrt{n}(\hat\psi-\psi^{*}) \;\rightsquigarrow\; \mathcal N\!\big(0,\; V\big),
\qquad 
IF(O)= -G^{-1}\phi(O;\psi^{*},\eta^{*}),\quad
V = G^{-1}\Sigma (G^{-1})^\top.
\]
\end{enumerate}
\end{theorem}
\begin{proof}
See Appendix \ref{appendix:proof1}.
\end{proof}

Algorithm \ref{algorithm1} computes the estimator from Theorem \ref{theorem1}.
Note that the theorem requires only that $T_g^\dagger$ satisfies the density-weighted adjoint identity (\ref{eq:adj_identity}).
The familiar formulas for the settings where $g$ is a continuous bijection or discrete are convenient instances but not required for the theorem to hold. 

%==================== Algorithm 1: Point Exposure ====================
\begin{algorithm}[H]
\caption{Cross-fitted Neyman–orthogonal estimator for a point-exposure MTP}
\begin{algorithmic}[1]
\Require Data $\{(Y_i,A_i,H_i)\}_{i=1}^n$; basis $s(H)$; shift $\delta$; folds $K$
\Ensure $\hat\psi$ and IF-sandwich covariance $\widehat V$
\State Partition indices into $K$ folds $(\mathcal T^{(k)},\mathcal I^{(k)})$
\State Initialize $M\gets 0$, $b\gets 0$
\For{$k=1,\dots,K$}
  \Comment{Nuisance fits on train fold}
  \State Fit $\hat{\mu}$ on $\{(Y_j,A_j,H_j)\}_{j\in\mathcal T^{(k)}}$ for $y\approx E[Y\mid A,H]$
  \State Fit $\hat m$ on $\{(A_j,H_j)\}_{j\in\mathcal T^{(k)}}$ for $a\approx E[A\mid H]$
  \State Estimate $\hat\sigma$ as SD of residuals $A-\hat m(H)$ on $\mathcal T^{(k)}$ (if using Normal ratio)
  \ForAll{$i\in\mathcal I^{(k)}$}
    \State $e_i \gets Y_i - \hat\mu(A_i,H_i)$;\ \ $\mu^g_i \gets \hat\mu(A_i+\delta,H_i)$
    \State $s_i \gets s(H_i)$;\ \ $q_i \gets s_i$
    \If{Normal working residual for $A\!\mid\!H$}
       \State $\hat r_i \gets \dfrac{\phi(A_i-\delta;\ \hat m(H_i),\hat\sigma)}{\phi(A_i;\ \hat m(H_i),\hat\sigma)}$
    \Else
       \State Estimate $\hat r_i \approx \dfrac{\hat f_{A\mid H}(A_i-\delta\mid H_i)}{\hat f_{A\mid H}(A_i\mid H_i)}$
    \EndIf
    \State $\tilde q_i \gets q_i \cdot \hat r_i$
    \Comment{Accumulate normal equations $M\hat\psi=b$}
    \State $M \gets M + \delta\, s_i s_i^\top$
    \State $b \gets b - \big[s_i\{\hat\mu(A_i,H_i)-\mu^g_i\} + (q_i-\tilde q_i)\,e_i\big]$
  \EndFor
\EndFor
\State \textbf{Solve and variance:}
\State $\hat\psi \gets (M+\lambda I)^{-1} b$ \Comment{optional small ridge $\lambda\ge 0$}
\State $G \gets (\delta/n)\sum_{i=1}^n s_i s_i^\top$
\State $\phi_i \gets (q_i-\tilde q_i)e_i + s_i\{\hat\mu(A_i,H_i)-\hat\mu(A_i+\delta,H_i) + \delta\, s_i^\top \hat\psi\}$
\State $\mathrm{IF}_i \gets G^{-1}\phi_i$;\ \ $\widehat V \gets \mathrm{Var}_n(\mathrm{IF}_i)/n$
\end{algorithmic}\label{algorithm1}
\end{algorithm}

\subsection{Time varying treatments}
For didactic purposes, consider identification in the two time point case. We extend the results by induction to the general time varying case in the Appendix. We know that the one time step ahead effects $\gamma_{12}^g(a_1,h_1)$ and $\gamma_{01}^g(a_0,l_0)$ are identified from the point exposure case. So we will focus on identification of $\gamma_{02}^g(a_0,l_0)$.
\begin{align*}
&\gamma_{02}^g(a_0,l_0)=E[Y_2(a_0,g_1)-Y_2(g_0,g_1)|L_0=l_0,A_0=a_0]
\end{align*}
By iterated expectations and Lemma \ref{lemma1},
\begin{align*}
&E[Y_2(a_0,g_1)|A_0=a_0,L_0=l_0]\\
&=E[E[Y_2(a_0,g_1)|A_0=a_0,L_0=l_0,L_1,A_1]|A_0=a_0,L_0=l_0]\\
&=E[E[Y_2-\gamma_{12}^g(A_1,H_1)|A_0=a_0,L_0=l_0,L_1,A_1]|A_0=a_0,L_0=l_0]\\
&=E[Y_2-\gamma_{12}^g(A_1,H_1)|A_0=a_0,L_0=l_0]
\end{align*}
Thus, the first term of $\gamma_{02}^g(a_0,l_0)$ is identified. The second term is identified as
\begin{align*}
&E[Y_2(g_0,g_1)|L_0=l_0,A_0=a_0]=\\
&\sum_{l_1,a_1}E[Y_2|L_0=l_0,A_0=g(l_0,a_0),L_1=l_1,A_1=g(l_0,g(l_0,a_0),a_1)]\\&p(l_1|L_0=l_0,A_0=g(l_0,a_0))p(a_1|L_0=l_0,L_1=l_1,A_0=g(l_0,a_0))\\
&\equiv\mu_{02}(l_0,g(l_0,a_0))
\end{align*}
by the extended g-formula under the assumptions of Section 5 of \citep{richardson2013single}. Following \citet{diaz2023nonparametric}, $\mu_{02}(l_0,a_0)$ can be represented in terms of iterative regressions by 
\begin{equation}
    \mu_{02}(l_0,a_0)=E[\mu_{12}(l_0,L_1,g(l_0,a_0),g(l_0,L_1,g(l_0,a_0),A_1))|L_0=l_0,A_0=a_0]
\end{equation}
with $\mu_{12}(l_0,l_1,a_0,a_1)=E[Y_2|L_0=l_0,L_1=l_1,A_0=a_0,A_1=a_1]$. Then for parametric model $(\gamma_{01}^g(l_0,a_0), \gamma_{12}^g(h_1,a_1),\gamma_{02}^g(l_0,a_0))=(\gamma_{01}^g(l_0,a_0;\psi^{*}), \gamma_{12}^g(h_1,a_1;\psi^{*}),\gamma_{02}^g(l_0,a_0;\psi^{*}))$, parameter $\psi^*$ satisfies
\begin{align}
    E\begin{bmatrix}
q_{01}(L_0,A_0)\{Y_1-\gamma_{01}^g(L_0,A_0;\psi^*)-\mu_{01}(L_0,g(L_0,A_0))\}\\
q_{12}(H_1,A_1)\{Y_1-\gamma_{12}^g(H_1,A_1;\psi^*)-\mu_{12}(H_1,g(H_1,A_1))\}\\
q_{02}(L_0,A_0)\{Y_1-\gamma_{12}^g(H_1,A_1;\psi^*)-\gamma_{02}^g(L_0,A_0;\psi^*)-\mu_{02}(L_0,g(L_0,A_0))\}
    \end{bmatrix}=0,
\end{align}
which again leads to natural estimating equations. In particular, a consistent and asymptotically normal estimator of $\psi^*$ can be obtained by solving
\begin{equation}
    \mathbb{P}_N\begin{bmatrix}b_{01}(L_0,A_0)\{Y_1-\mu_{01}(L_0,A_0;\beta)\\
    b_{12}(\bar{L}_1,\bar{A}_1)\{Y_2-\mu_{12}(\bar{L}_1,\bar{A}_1;\beta)\\
    b_{02}(L_0,A_0)\{\mu_{12}(L_0,L_1,g(L_0,A_0),g(L_0,L_1,g(L_0,A_0),A_1))-\mu_{02}(L_0,A_0;\beta)\}\\
    q_{01}(L_0,A_0)\{Y_1-\gamma_{01}^g(L_0,A_0;\psi)-\mu_{01}(L_0,g(L_0,A_0))\}\\
q_{12}(H_1,A_1)\{Y_1-\gamma_{12}^g(H_1,A_1;\psi)-\mu_{12}(H_1,g(H_1,A_1))\}\\
q_{02}(L_0,A_0)\{Y_1-\gamma_{12}^g(H_1,A_1;\psi)-\gamma_{02}^g(L_0,A_0;\psi)-\mu_{02}(L_0,g(L_0,A_0))\}
    \end{bmatrix}=0
\end{equation}
where $\mu_{mk}(\bar{A}_m,\bar{L}_m;\beta)$ are models for $\mu_{mk}(\bar{A}_m,\bar{L}_m)$ smooth in an $r$-dimensional parameter $\beta$ (equal to $\beta^*$ under the true law) and $q_{mk}(\bar{A}_m,\bar{L}_m)$ and $b_{mk}(\bar{A}_m,
\bar{L}_m)$ are conformable analyst selected index functions.

As in the point exposure setting, a Neyman orthogonal estimator would be desirable to enable machine learning estimation of nuisance functions. We now develop this estimator for the general time-varying setting. Define the one step ahead conditional mean
\[
\mu_t(a,h)\equiv E[V_{t+1}|A_t=a,H_t=h]
\]
where we set $V_T = Y$ and recursively set for $t=T-1,\ldots,0$
\[
V_t = \mu_t(g(H_t,A_t),H_t).
\]
Now, for each time $t$, define the operator $T_{gt}h\equiv h(H_t,g(H_t,A_t))$. And let $T_{gt}^{\dagger}$ be its adjoint with respect to
\[
\langle u,v\rangle_{w_t} \equiv E[u(H_t,A_t)v(H_t,A_t)w_t]
\]
for $w_t=\pi(A_t|H_t)$. As in the point exposure setting, let $\tilde{q}_t \equiv T_{gt}^{\dagger}q_t$.

\begin{theorem}[Neyman orthogonal estimator, time-varying]\label{theorem2}
Assume Sequential consistency, positivity, and MTP exchangeability (i.e. the assumptions of Section 5 of \citet{richardson2013single}) hold. Consider the score
\[
\Phi(O;\psi,\eta)=\sum_{t=0}^{T-1}\phi_t(O;\psi,\eta),
\]
with
\[
\phi_t
=\big(q_t-\tilde q_t\big)\,\{V_{t+1}-\mu_t(H_t,A_t)\}
+q_t\Big\{\mu_t(H_t,A_t)-\mu_t(H_t,g_t(A_t,H_t))-\gamma_t(H_t,A_t;\psi)\Big\},
\]
where $\eta=\{\mu_t,\pi_t,\tilde q_t: t=0,\dots,T-1\}$.  
\begin{enumerate}
    \item \textbf{Identification} At the truth $(\psi^\ast,\eta^\ast)$, $\mathbb E[\Phi(O;\psi^\ast,\eta^\ast)]=0$
    \item \textbf{Neyman orthogonality} The score is Neyman-orthogonal: the pathwise derivative of $\mathbb E[\Phi(O;\psi^\ast,\eta)]$ with respect to each nuisance $\mu_t$ and $\pi_t$ vanishes at $\eta^\ast$.
    \item \textbf{Asymptotic variance} Suppose $\psi$ is estimated by the root $\hat\psi$ of the cross-fitted sample moment $n^{-1}\sum_{i=1}^n \hat\phi_i(\psi)=0$, with nuisances learned on held-out folds and consistent in $L_2$. 
Let
\[
G:=E\!\left[\frac{\partial}{\partial\psi^\top}\phi(O;\psi^\ast,\eta^\ast)\right],\qquad
\Sigma:=Var\!\big(\phi(O;\psi^\ast,\eta^\ast)\big).
\]
Then under standard regularity conditions,
\[
\sqrt{n}(\hat\psi-\psi^\ast)\;\overset{d}{\to}\; \mathcal N\!\big(0,\; V\big),
\qquad 
\text{with IF } \ IF(O)= -\,G^{-1}\phi(O;\psi^\ast,\eta^\ast),
\]
and
\[
V \;=\; G^{-1}\,\Sigma\,(G^{-1})^\top.
\]
\end{enumerate}
\end{theorem}
\begin{proof}
    See Appendix \ref{appendix:proof2}
\end{proof}

Algorithm \ref{algorithm2} implements the estimator from Theorem \ref{theorem2}.

%==================== Algorithm 2: Time-Varying ====================
\begin{algorithm}[H]
\caption{Cross-fitted Neyman–orthogonal estimator for longitudinal MTPs}
\begin{algorithmic}[1]
\Require Data $\{(Y_i,\{A_{t,i},H_{t,i}\}_{t=0}^{T-1})\}_{i=1}^n$; bases $s_t(H_t)$; shifts $\{\delta_t\}$; folds $K$
\Ensure $\hat\psi=(\hat\psi_0^\top,\dots,\hat\psi_{T-1}^\top)^\top$ and IF-sandwich $\widehat V$
\State Initialize $V_{T,i}\gets Y_i$ for all $i$;\quad $M_t\gets 0$, $b_t\gets 0$ for $t=0,\dots,T\!-\!1$
\State Partition indices into $K$ folds $(\mathcal T^{(k)},\mathcal I^{(k)})$
\For{$t=T-1,\dots,0$}
  \For{$k=1,\dots,K$}
    \Comment{Nuisance fits on train fold at time $t$}
    \State Fit $\hat\mu_t$ on $\{(V_{t+1,j},A_{t,j},H_{t,j})\}_{j\in\mathcal T^{(k)}}$ for $v\approx E[V_{t+1}\mid A_t,H_t]$
    \State Fit $\hat m_t$ on $\{(A_{t,j},H_{t,j})\}_{j\in\mathcal T^{(k)}}$ for $a\approx E[A_t\mid H_t]$
    \State Estimate $\hat\sigma_t$ as SD of $A_t-\hat m_t(H_t)$ on $\mathcal T^{(k)}$ (if using Normal ratio)
    \ForAll{$i\in\mathcal I^{(k)}$}
      \State $e_{t,i} \gets V_{t+1,i} - \hat\mu_t(A_{t,i},H_{t,i})$;\ \ $\mu^g_{t,i} \gets \hat\mu_t(A_{t,i}+\delta_t,H_{t,i})$
      \State $s_{t,i} \gets s_t(H_{t,i})$;\ \ $q_{t,i}\gets s_{t,i}$
      \If{Normal working residual}
         \State $\hat r_{t,i} \gets \dfrac{\phi(A_{t,i}-\delta_t;\ \hat m_t(H_{t,i}),\hat\sigma_t)}{\phi(A_{t,i};\ \hat m_t(H_{t,i}),\hat\sigma_t)}$
      \Else
         \State Estimate $\hat r_{t,i} \approx \dfrac{\hat f_{A_t\mid H_t}(A_{t,i}-\delta_t\mid H_{t,i})}{\hat f_{A_t\mid H_t}(A_{t,i}\mid H_{t,i})}$
      \EndIf
      \State $\tilde q_{t,i} \gets q_{t,i}\cdot \hat r_{t,i}$
      \Comment{Accumulate time-$t$ normal equations}
      \State $M_t \gets M_t + \delta_t\, s_{t,i}s_{t,i}^\top$
      \State $b_t \gets b_t - \big[s_{t,i}\{\hat\mu_t(A_{t,i},H_{t,i})-\mu^g_{t,i}\} + (q_{t,i}-\tilde q_{t,i})\,e_{t,i}\big]$
      \State $V_{t,i} \gets \mu^g_{t,i}$ \Comment{backward recursion target}
    \EndFor
  \EndFor
\EndFor
\State \textbf{Solve and variance:}
\For{$t=0,\dots,T-1$} \State $\hat\psi_t \gets (M_t+\lambda_t I)^{-1} b_t$ \EndFor
\State $G \gets \mathrm{diag}\!\big((\delta_0/n)\sum_i s_{0,i}s_{0,i}^\top,\dots,(\delta_{T-1}/n)\sum_i s_{T-1,i}s_{T-1,i}^\top\big)$
\State For each $i$: $\phi_{t,i} \gets (q_{t,i}-\tilde q_{t,i})e_{t,i} + s_{t,i}\{\hat\mu_t(A_{t,i},H_{t,i})-\hat\mu_t(A_{t,i}+\delta_t,H_{t,i}) + \delta_t\, s_{t,i}^\top \hat\psi_t\}$
\State Stack $\phi_i\gets(\phi_{0,i}^\top,\dots,\phi_{T-1,i}^\top)^\top$;\ \ $\mathrm{IF}_i \gets G^{-1}\phi_i$;\ \ $\widehat V \gets \mathrm{Var}_n(\mathrm{IF}_i)/n$
\end{algorithmic}\label{algorithm2}
\end{algorithm}

\section{Identification and Estimation Under Parallel Trends Assumptions}\label{section_pt}
Recall that in this section $\bar{L}_t$ excludes the most recent outcome $Y_t$. In the parallel trends setting, we present identification results and the natural accompanying parametric outcome regression based estimators. We leave derivation of Neyman orthogonal estimators under parallel trends for future work.
\subsection{Point Exposure Setting}
Again, we suppress time subscripts for treatment and the baseline covariate to reduce notational clutter in the point exposure setting. We propose the MTP parallel trends assumption
\begin{equation}\label{pe_pt}
    E[Y_1(g(a,l))-Y_0|A=a,L=l]=E[Y_1(g(a,l))-Y_0|A=g(a,l),L=l] \text{ for all } (a,l) \text{ with } p(a,l)>0.
\end{equation}
In words, the conditional expected counterfactual trend in the outcome under the MTP does not depend on whether treatment takes its natural or modified value. Under (\ref{pe_pt}), $\gamma^g(l,a)$ is identified as follows:
\begin{align*}
&\gamma^g(l,a)=E[Y_1-Y_1(g(a,l))|A=a,L=l]\\
    &=E[Y_1-Y_0|A=a,L=l]-E[Y_1(g(a,l))-Y_0|A=a,L=l]\\
    &=E[Y_1-Y_0|A=a,L=l]-E[Y_1(g(a,l))-Y_0|A=g(a,l),L=l]\\
    &=E[Y_1-Y_0|A=a,L=l]-E[Y_1-Y_0|A=g(a,l),L=l]
\end{align*}
Thus, MTPs can be studied under parallel trends assumptions. Furthermore, $\psi^*$ from (\ref{snmm_par}) satisfies
\begin{equation}
    E[q(L,A)(Y_1-Y_0-\gamma^g(L,A;\psi^*)-\mu^d(L,g(A,L)))]=0,
\end{equation}
with $\mu^d(l,a)=E[Y_1-Y_0|L=l,A=a]$, which again suggests a natural estimation procedure. A consistent and asymptotically normal estimator of $\psi^*$ can be obtained by solving the corresponding estimating equations 
\begin{equation} \mathbb{P}_N\begin{bmatrix}b(A,L)\{Y_1-Y_0-\mu^d(L,A;\beta)\}\\
q(A,L)\{Y_1-Y_0-\gamma^g(A,L;\psi)-\mu^d(L,A;\beta)\}\end{bmatrix}=0
\end{equation}
where $\mu^d(L,A;\beta)$ is a model for $\mu^d(L,A)$ smooth in an $r$-dimensional parameter $\beta$ (equal to $\beta^*$ under the true law) and $q(A,L)$ and $b(A,L)$ are conformable analyst selected index functions. 

\subsection{Time-Varying Treatments}
In the time-varying setting, we make the parallel trends assumptions
\begin{equation}
    E[Y_k(g)-Y_{k-1}(g)|A_m=a_m,H_m=h_m]=E[Y_k(g)-Y_{k-1}(g)|A_m=g(h_m,a_m),H_m=h_m] \text{ for all }h_m, a_m, k>m.
\end{equation}
In words, conditional expected future counterfactual trends under $g$ from time $m$ onwards do not depend on whether treatment at time $m$ took its natural or modified value. For the two time point setting, these assumptions can be enumerated
\begin{align}
\begin{split}
    E[Y_2(g)-Y_1(g(a_0,l_0))|A_0=a_0,L_0=l_0]=E[Y_2(g)-Y_1(g(l_0,a_0))|A_0=g(l_0,a_0),L_0=l_0]\\
    E[Y_1(g(l_0,a_0))-Y_0|A_0=a_0,L_0=l_0]=E[Y_1(g(l_0,a_0))-Y_0|A_0=g(l_0,a_0),L_0=l_0]\\
    E[Y_2(g)-Y_1|A_1=a_1,H_1=h_1]=E[Y_2(g)-Y_1|A_1=g(h_1,a_1),H_1=h_1]
    .
    \end{split}
\end{align}
We know that the one time step ahead effects $\gamma_{12}^g(a_1,h_1)$ and $\gamma_{01}^g(a_0,l_0)$ are identified from the point exposure case. Identification of $\gamma_{02}^g(a_0,l_0)$ can be demonstrated as follows. 
\begin{align*}
    &\gamma_{02}(a_0,l_0)=E[Y_2(a_0,g)-Y_2(g)|A_0=a_0,L_0=l_0]\\
    &=E[(Y_2(a_0,g)-Y_1(g))-(Y_2(g)-Y_1(g))|A_0=a_0,L_0=l_0]
\end{align*}
The first term of the conditional expectation on the right hand side is identified as
\begin{align*}
    E[Y_2(a_0,g)-Y_1(g)|A_0=a_0,L_0=l_0]=E[Y_2-\gamma_{12}(l_0,L_1,a_0,A_1)-(Y_1-\gamma_{01}(l_0,a_0))|A_0=a_0,L_0=l_0].
\end{align*}
The second term is identified as
\begin{align*}
   & E[Y_2(g)-Y_1(g)|A_0=a_0,L_0=l_0]\\
   & =E[Y_2(g)-Y_1(g)|A_0=g(a_0,l_0),L_0=l_0]\\
&=\int_{l_1,a_1}E[Y_2(g)-Y_1|A_0=g(a_0,l_0),A_1=a_1,L_0=l_0,L_1=l_1]p(l_1,a_1|A_0=g(a_0,l_0),L_0=l_0)da_1dl_1\\
    &=\int_{l_1,a_1}E[Y_2-Y_1|A_0=g(a_0,l_0),A_1=g(l_0,l_1,g(l_0,a_0),a_1),L_0=l_0,L_1=l_1]p(l_1,a_1|A_0=g(a_0,l_0),L_0=l_0)da_1dl_1\\
    &=E[\mu_{12}^d(l_0,L_1,g(a_0,l_0),g(l_0,L_1,g(l_0,a_0),A_1)|A_0=g(l_0,a_0),L_0=l_0]\\
    &\equiv\mu_{02}^d(g(l_0,a_0),l_0)
\end{align*}
where $\mu_{12}^d(l_0,l_1,a_0,a_1)=E[Y_2-Y_1|L_0=l_0,L_1=l_1,A_0=a_0,A_1=a_1]$. The derivation is simply
repeated applications of iterated expectations, parallel trends, and consistency. It follows that for parametric model $(\gamma_{01}^g(l_0,a_0), \gamma_{12}^g(h_1,a_1),\gamma_{02}^g(l_0,a_0))=(\gamma_{01}^g(l_0,a_0;\psi^{*}), \gamma_{12}^g(h_1,a_1;\psi^{*}),\gamma_{02}^g(l_0,a_0;\psi^{*}))$, parameter $\psi^*$ satisfies
\begin{align}
    E\begin{bmatrix}
q_{01}(L_0,A_0)\{Y_1-Y_0-\gamma_{01}^g(L_0,A_0;\psi^*)-\mu^d_{01}(L_0,g(L_0,A_0))\}\\
q_{12}(H_1,A_1)\{Y_2-Y_1-\gamma_{12}^g(H_1,A_1;\psi^*)-\mu^d_{12}(H_1,g(H_1,A_1))\}\\
q_{02}(L_0,A_0)\{Y_2-Y_1-\gamma_{12}^g(H_1,A_1;\psi^*)-\gamma_{02}^g(L_0,A_0;\psi^*)+\gamma_{01}^g(L_0,A_0;\psi^*)-\mu_{02}^d(L_0,g(L_0,A_0))\}
    \end{bmatrix}=0,
\end{align}
which again leads to natural estimating equations. In particular, a consistent and asymptotically normal estimator of $\psi^*$ can be obtained by solving
\begin{equation}
    \mathbb{P}_N\begin{bmatrix}b_{01}(L_0,A_0)\{Y_1-Y_0-\mu^d_{01}(L_0,A_0;\beta)\\
    b_{12}(\bar{L}_1,\bar{A}_1)\{Y_2-Y_1-\mu^d_{12}(\bar{L}_1,\bar{A}_1;\beta)\\
    b_{02}(L_0,A_0)\{\mu^d_{12}(L_0,L_1,g(L_0,A_0),g(L_0,L_1,g(L_0,A_0),A_1))-\mu^d_{02}(L_0,A_0;\beta)\}\\
    q_{01}(L_0,A_0)\{Y_1-Y_0-\gamma_{01}^g(L_0,A_0;\psi)-\mu^d_{01}(L_0,g(L_0,A_0))\}\\
q_{12}(H_1,A_1)\{Y_2-Y_1-\gamma_{12}^g(H_1,A_1;\psi)-\mu^d_{12}(H_1,g(H_1,A_1))\}\\
q_{02}(L_0,A_0)\{Y_2-Y_1-\gamma_{12}^g(H_1,A_1;\psi)-\gamma_{02}^g(L_0,A_0;\psi)+\gamma_{01}^g(L_0,A_0;\psi)-\mu_{02}^d(L_0,g(L_0,A_0))\}
    \end{bmatrix}=0
\end{equation}
where $\mu_{mk}^d(\bar{A}_m,\bar{L}_m;\beta)$ are models for $\mu_{mk}^d(\bar{A}_m,\bar{L}_m)$ smooth in an $r$-dimensional parameter $\beta$ (equal to $\beta^*$ under the true law) and $q_{mk}(\bar{A}_m,\bar{L}_m)$ and $b_{mk}(\bar{A}_m,
\bar{L}_m)$ are conformable analyst selected index functions. 

\subsection{Neyman orthogonal estimators}
To derive Neyman orthogonal estimators enabling machine learning estimation of nuisance functions, we proceed completely analogously to the (sequential) exchangeability case. In fact, we can simply substitute the difference $\Delta Y_t\equiv Y_t-Y_{t-1}$ for $Y_t$ and the regressions $\mu_{mk}^d(H_m,A_m)=E[\Delta Y_k|H_m,A_m]$ for $\mu_{mk}(H_m,A_m)$ in the scores from Theorems \ref{theorem1} and \ref{theorem2}. The proofs of those theorems then go through exactly as before. Here, we state the theorem for the general longitudinal parallel trends case.

\begin{theorem}[Neyman–orthogonal estimator under parallel trends, longitudinal]\label{thm:PT-long}
Set
\[
\mu_t^d(h_t,a_t):=\E[\Delta Y_t\mid H_t=h_t,A_t=a_t].
\]
For each $t$, equip functions with
\(
\langle u,v\rangle_t:=\E\!\big[\int u(H_t,a_t)v(H_t,a_t) f_{A_t\mid H_t}(a_t\mid H_t)\,da_t\big].
\)
Let $T_{g_t}$ act by $(T_{g_t} q_t)(h_t,a_t)=q_t(h_t,g_t(h_t,a_t))$ with adjoint $T_{g_t}^\dagger$ under $\langle\cdot,\cdot\rangle_t$.
Choose $q_t(H_t,A_t):=s_t(H_t)$ and set $\tilde q_t:=T_{g_t}^\dagger q_t$.
Let $\Gamma_t(H_t;\psi_t)$ be linear in $s_t(H_t)$.

Assume for each $t$: $\E[\Delta Y_t(0)\mid H_t,A_t]=\E[\Delta Y_t(0)\mid H_t]$. Further assume sequential consistency, positivity, and regularity conditions.

Define the time-$t$ score
\[
\phi_t
=(q_t-\tilde q_t)\{\Delta Y_t-\mu_t^d(H_t,A_t)\}
+ s_t(H_t)\{\mu_t^d(H_t,A_t)-\mu_t^d(H_t,g_t(H_t,A_t)) + \Gamma_t(H_t;\psi_t)\}.
\]
Let $\phi=(\phi_0^\top,\dots,\phi_{T-1}^\top)^\top$ and $\psi=(\psi_0^\top,\dots,\psi_{T-1}^\top)^\top$.

\noindent Then:
\begin{enumerate}\itemsep4pt
\item[(i)] \textbf{Identification.} For each $t$, there is a unique $\psi_t^*$ s.t.
\(
\E[s_t\{\mu_t^d(H_t,A_t)-\mu_t^d(H_t,g_t(H_t,A_t))+\Gamma_t(H_t;\psi_t)\}]=0
\),
and $\E[\phi_t(\psi_t^*,\eta_t^*)]=0$. Hence $\E[\phi(\psi^*,\eta^*)]=0$.

\item[(ii)] \textbf{Neyman orthogonality.}
For any regular paths $\eta_{t,\varepsilon}=(\mu_{t,\varepsilon}^d,f_{t,\varepsilon})$ through $\eta_t^*$,
\(
\left.\frac{d}{d\varepsilon}\E[\phi_t(\psi_t^*,\eta_{t,\varepsilon})]\right|_{\varepsilon=0}=0
\)
for each $t$.

\item[(iii)] \textbf{Asymptotic normality.}
Suppose $\psi$ is estimated by the root $\hat\psi$ of the cross-fitted sample moment $n^{-1}\sum_{i=1}^n \hat\phi_i(\psi)=0$, with nuisances learned on held-out folds and consistent in $L_2$. 
Let
\[
G:=E\!\left[\frac{\partial}{\partial\psi^\top}\phi(O;\psi^\ast,\eta^\ast)\right],\qquad
\Sigma:=Var\!\big(\phi(O;\psi^\ast,\eta^\ast)\big).
\]
Then under standard regularity conditions,
\[
\sqrt{n}(\hat\psi-\psi^\ast)\;\overset{d}{\to}\; \mathcal N\!\big(0,\; V\big),
\qquad 
\text{with IF } \ IF(O)= -\,G^{-1}\phi(O;\psi^\ast,\eta^\ast),
\]
and
\[
V \;=\; G^{-1}\,\Sigma\,(G^{-1})^\top.
\]
\end{enumerate}
\end{theorem}

\section{Simulations}\label{sec:sim}
\subsection{Point exposure}
We simulated data according to the data generating process:
\begin{align*}
    &L\sim N(0,1)\\
    &A|L \sim N(m(L),1) \text{ with }m(L)=\theta_0 + \theta_1 L + \theta_2 L^2\\
    &Y = \xi_0 +\xi_1L+\xi_2L^2+(\beta_0+\beta_1L)A+\epsilon;\text{ }\epsilon\sim N(0,1)
\end{align*}
We then sought to estimate the heterogeneous effects of the shift MTP $g(A,L)=A+\delta$. Are interested in the SNMM
\begin{equation*}
    \gamma(A,L) = E[Y|A,L] - E[Y|A+\delta,L] = -\delta(\beta_0+\beta_1)L 
\end{equation*}
Thus, we specify the correct parametric SNMM
\[
\gamma(A,L;\psi)=-\delta s(L)^T\psi
\]
with $s(L)=(1,L)$ and $\psi=(\beta_0,\beta1)$. We then define 
\[
\tilde{q}(a,l)=q(a-\delta,l)\frac{p(a-\delta|l)}{p(a|l)}
\]
and choose $q(A,L)=s(L)$. We set true parameter values to be:
\begin{align*}
    \theta = (0.2,0.8,-0.4); \text{  }\xi=(0.3,0.5,0.2);\text{  }\beta=(0.8,-0.6).
\end{align*}
We then estimated the effect of a shift by $\delta=0.5$ by solving the estimating equations of Theorem \ref{theorem1} with cross-fitting for data sets of sizes 400, 1000, and 3000. We used correctly specified parametric nuisance models. Table \ref{tab:pe_sim_results} shows that the estimator is unbiased and that sandwich confidence intervals had nominal coverage for each component of $\psi$. 

\begin{table}[htbp]
\centering
\caption{Monte Carlo performance — point-exposure continuous-shift orthogonal estimator (stable score; $R=200$).}
\label{tab:pe_sim_results}
\setlength{\tabcolsep}{4pt}
\sisetup{table-number-alignment=center}
\begin{tabular}{S[table-format=4.0] *{10}{S[table-format=1.3]}}
\toprule
& \multicolumn{5}{c}{$\psi_{0}$} & \multicolumn{5}{c}{$\psi_{1}$} \\
\cmidrule(lr){2-6}\cmidrule(lr){7-11}
{n} & {Bias} & {RMSE} & {EmpSD} & {Mean SE} & {Cov95} & {Bias} & {RMSE} & {EmpSD} & {Mean SE} & {Cov95} \\
\midrule
400  & 0.002 & 0.056 & 0.056 & 0.057 & 0.96 & 0.004 & 0.065 & 0.065 & 0.059 & 0.93 \\
1000 & 0.001 & 0.037 & 0.037 & 0.035 & 0.93 & 0.002 & 0.037 & 0.037 & 0.035 & 0.93 \\
3000 & 0.000 & 0.020 & 0.020 & 0.020 & 0.94 & -0.002 & 0.020 & 0.019 & 0.020 & 0.96 \\
\bottomrule
\end{tabular}

\begin{minipage}{0.92\linewidth}
\footnotesize
\textit{Notes:} Bias = mean of $\hat\psi-\psi^\star$; EmpSD = empirical SD of $\hat\psi$ across replications; 
Mean SE = average sandwich SE; Cov95 = empirical coverage of 95\% Wald CIs from sandwich SEs.
\end{minipage}
\end{table}

\subsection{Time-varying treatment}

We consider a two–time-point longitudinal setting with baseline covariate $L_0$, treatments $A_0,A_1\in\mathbb{R}$, intermediate covariate $L_1$, and outcome $Y\in\mathbb{R}$. The modified treatment policy (MTP) shifts each treatment by a constant
\[
g_0(a_0)=a_0+\delta_0,\qquad g_1(a_1)=a_1+\delta_1,
\]
with fixed $\delta_0,\delta_1>0$. The blip functions relative to $g=(g_0,g_1)$ are parameterized by low–dimensional bases
\[
s_0(H_0)=(1,L_0)^\top,\qquad s_1(H_1)=(1,L_1)^\top,
\]
and true parameters $\psi_0=(\psi_{0,0},\psi_{0,1})^\top$, $\psi_1=(\psi_{1,0},\psi_{1,1})^\top$, so that
\[
\gamma_0^g(H_0,A_0)\;=\;-\delta_0\,s_0(H_0)^\top\psi_0,\qquad
\gamma_1^g(H_1,A_1)\;=\;-\delta_1\,s_1(H_1)^\top\psi_1.
\]

\paragraph{Observed data–generating process (DGP).}
We calibrated our generating process such that we could derive the true parameter values for our blip model. Let $n$ i.i.d.\ observations be generated as follows; all noises are mutually independent and independent of past history.
\begin{align*}
L_0 &\sim \mathcal{N}(0,1),\\
A_0 \mid L_0 &\sim \mathcal{N}\!\big(\,0.4\,L_0,\; \sigma_{A_0}^2\,\big),\\
L_1 \mid (L_0,A_0) &:= \rho_0 + \rho_1 L_0 + \rho_2 A_0 + \nu,\qquad \nu\sim \mathcal{N}(0,\sigma_{L_1}^2),\\
A_1 \mid (L_0,L_1,A_0) &\sim \mathcal{N}\!\big(\,\kappa_0 + \kappa_2 L_0,\; \sigma_{A_1}^2\,\big),\\[2pt]
\mu_1(H_1,a_1) &:= b_1(L_0,L_1) + \{\psi_{1,0}+\psi_{1,1}L_1\}\,a_1,\\
Y \mid (H_1,A_1) &:= \mu_1(H_1,A_1) + \varepsilon,\qquad \varepsilon\sim \mathcal{N}(0,\sigma_Y^2),
\end{align*}
where $b_1(L_0,L_1)=\beta_{10}+\beta_{1L_1} L_1+\beta_{1L_0} L_0$ is a baseline outcome component. With this construction,
\[
\mu_1(H_1,a_1)-\mu_1\!\big(H_1,a_1+\delta_1\big)\;=\;-\delta_1\{\psi_{1,0}+\psi_{1,1}L_1\}\;=\;\gamma_1^g(H_1,A_1)
\]
holds by design for any $H_1$.

To ensure the time-$0$ blip identity also holds,
\[
\mu_0(H_0,a_0)-\mu_0\!\big(H_0,a_0+\delta_0\big)\;=\;-\delta_0\,s_0(H_0)^\top\psi_0\;=\;\gamma_0^g(H_0,A_0),
\]
we calibrate $\psi_0$ so that the slope in $a_0$ of
\[
\mathbb{E}\!\left[\,\mu_1\!\big(H_1,A_1+\delta_1\big)\,\middle|\,H_0,\ A_0=a_0\right]
\]
equals $s_0(H_0)^\top\psi_0$ for all $L_0$. Under the simple $A_1$ law above (mean depending on $L_0$ but not on $A_0$ or $L_1$), a short calculation yields
\[
\psi_{0,0}\;=\;\rho_2\,\beta_{1L_1}\;+\;\rho_2\,\psi_{1,1}\,(\kappa_0+\delta_1),\qquad
\psi_{0,1}\;=\;\rho_2\,\psi_{1,1}\,\kappa_2
\]
so that $\mu_0(H_0,a_0):=\mathbb{E}\!\left[\,\mu_1\!\big(H_1,A_1+\delta_1\big)\,\middle|\,H_0,\ A_0=a_0\right]$ is linear in $a_0$ with slope $s_0(H_0)^\top\psi_0$.

\paragraph{Default parameter values.}
In all experiments we use
\[
\delta_0=0.4,\quad \delta_1=0.5,\quad
\psi_1=(0.5,\,0.3),\quad
(\rho_0,\rho_1,\rho_2)=(0.1,\,0.6,\,0.8),
\]
\[
(\kappa_0,\kappa_2)=(0.2,\,0.35),\quad
(\beta_{10},\beta_{1L_1},\beta_{1L_0})=(0.25,\,0.5,\,0.2),
\]
and standard deviations $(\sigma_{A_0},\sigma_{L_1},\sigma_{A_1},\sigma_Y)=(1.0,\,0.5,\,1.0,\,1.0)$. The calibration then implies the \emph{true} time-$0$ blip coefficients
\[
\psi_{0,0}=0.8\times 0.5+0.8\times 0.3\times(0.2+0.5)=0.568,\qquad
\psi_{0,1}=0.8\times 0.3\times 0.35=0.084.
\]
Thus the overall target parameter vector is
\[
\psi^\star=\big(\psi_{0,0},\psi_{0,1},\psi_{1,0},\psi_{1,1}\big)^\top
=\big(0.568,\ 0.084,\ 0.5,\ 0.3\big)^\top.
\]

\paragraph{Results}
Table \ref{tab:tv_sim_results} displays the bias, bootstrap coverage, and empirical standard errors over 500 data sets of size n=1,000. We observe very low bias and near nominal coverage.

\begin{table}[htbp]
\centering
\caption{Monte Carlo (R=500) for longitudinal MTP estimator with continuous shifts; bootstrap 95\% coverage reported.}
\label{tab:tv_sim_results}
\setlength{\tabcolsep}{4pt}
\sisetup{table-number-alignment=center}
\begin{tabular}{S[table-format=4.0] *{12}{S[table-format=1.3]}}
\toprule
& \multicolumn{3}{c}{$\psi_{0,0}$} & \multicolumn{3}{c}{$\psi_{0,1}$} & \multicolumn{3}{c}{$\psi_{1,0}$} & \multicolumn{3}{c}{$\psi_{1,1}$} \\
\cmidrule(lr){2-4}\cmidrule(lr){5-7}\cmidrule(lr){8-10}\cmidrule(lr){11-13}
{n} & {Bias} & {RMSE} & {EmpSD} & {Bias} & {RMSE} & {EmpSD} & {Bias} & {RMSE} & {EmpSD} & {Bias} & {RMSE} & {EmpSD} \\
\midrule
1000 & -0.000 & 0.049 & 0.049 & -0.002 & 0.039 & 0.039 & 0.000 & 0.033 & 0.033 & 0.002 & 0.026 & 0.026 \\
\addlinespace[2pt]
\multicolumn{1}{r}{\textit{Boot 95\% cov.}} &
\multicolumn{3}{c}{\textit{0.93}} &
\multicolumn{3}{c}{\textit{0.94}} &
\multicolumn{3}{c}{\textit{0.96}} &
\multicolumn{3}{c}{\textit{0.96}} \\
\bottomrule
\end{tabular}

\end{table}

\section{Real Data Application: Effect of Shifting Workplace Mobility on Covid Incidence}\label{section:application}

\paragraph{Data.}
We linked the Google Community Mobility Reports \citep{google_mobility_landing_2022} (workplaces index; daily \% point deviation from the Jan~3--Feb~6,~2020 baseline) with New York Times county-level COVID-19 case counts \citep{nyt_covid_counties}. Units were counties. The exposure $A$ is the 7-day average of the workplace mobility index ending at $t_0$ (units: percentage points change from baseline mobility in the January 3- February 6 period). We defined the calendar date $t_0$ to be the date with a mobility measurement closest to June 1, 2020, breaking ties arbitrarily. We excluded 74 counties without an eligible $t_0$. Outcomes are future incident cases per 100,000 over the 7 days beginning 14 days after $t_0$. 

\paragraph{Intervention (MTP).}
We consider a constant shift policy $g(a)=a+\delta$ with $\delta=-5$ percentage points. This means that under this intervention each county's percentage point change in mobility score from Jan3-Feb6 would be shifted down 5 percentage points from its observed value.

\paragraph{Adjustment covariates.}
We adjusted for: population, Rural--Urban Continuum (RUC) score (1=most urban, 9=most rural), percent Black, percent Hispanic, Republican vote share in 2016, poverty rate, unemployment rate, uninsurance rate, land area, population density, and per-capita income.

\paragraph{Estimand and blip specification.}
We model the blip as linear in the RUC score:
\[
\Delta(H)\equiv \mathbb{E}\{Y(A+\delta)-Y(A)\mid H\}
\approx \delta\, s(H)^\top \psi,\qquad s(H)=(1,\ \mathrm{RUC})^\top.
\]
Here, $\psi_1$ is the average effect per 1 p.p.\ shift at RUC$=0$ (intercept on the RUC scale) and $\psi_2$ captures linear effect modification by RUC. With $\delta=-5$, a positive $\psi_2$ means that a decreasing mobility is more protective in more rural areas. Treating RUC as a nominal variable when it is really ordinal is a slight abuse. It surely pales in comparison to the omission of other covariates from the blip function. Keep in mind that this is an illustrative analysis.

\paragraph{Estimator.}
We used the point exposure Neyman-orthogonal estimator under exchangeability assumptions from Theorem \ref{theorem1}. Nuisance functions were estimated with gradient-boosted trees (XGBoost): $\mu(a,h)\approx \mathbb{E}[Y\mid A=a,H=h]$ and $m(h)\approx \mathbb{E}[A\mid H=h]$, with a Normal working residual for $A\mid H$ to form the density-ratio pullback $p(a-\delta\mid h)/p(a\mid h)$. We used $5$-fold cross-fitting. Standard errors come from the influence-function (IF) sandwich; Wald CIs are reported. Pointwise CIs for effects at specific RUC values use the delta method, i.e., $\widehat{\mathrm{Var}}\{\delta\,s(H)^\top\hat\psi\}=\delta^2 s(H)^\top \widehat{\mathrm{Var}}(\hat\psi)\, s(H)$.

\paragraph{Results.}
Estimated blip coefficients (per 1 p.p.\ shift):
\[
\hat\psi_1=14.3\quad(\mathrm{SE}=14.0;\ 95\%\ \mathrm{CI}:\ -13.2,\ 41.8),\qquad
\hat\psi_2=-3.1\quad(\mathrm{SE}=5.0;\ 95\%\ \mathrm{CI}:\ -13.0,\ 6.8).
\]
With the policy shift $\delta=-5$ p.p., the implied change in cases per 100{,}000 at a given RUC is
$\widehat{\Delta}(\mathrm{RUC})=\delta\,(\hat\psi_1+\hat\psi_2\,\mathrm{RUC})$.
Delta-method estimates (and 95\% CIs) at representative RUC values are:
\[
\begin{array}{lccc}
\text{RUC}=1: & \widehat{\Delta}=-7.24,\ \mathrm{SE}=38.08,\ \mathrm{CI}\ [-81.88,\ 67.40];\\
\text{RUC}=5: & \widehat{\Delta}=-20.79,\ \mathrm{SE}=22.46,\ \mathrm{CI}\ [-64.82,\ 23.24];\\
\text{RUC}=9: & \widehat{\Delta}=-34.34,\ \mathrm{SE}=41.65,\ \mathrm{CI}\ [-115.97,\ 47.29].
\end{array}
\]

\paragraph{Interpretation.}
Point estimates suggest that a 5 p.p.\ reduction in workplace mobility index relative to what was observed would have lowered future COVID-19 cases per 100{,}000, with larger reductions in more rural counties (more negative at higher RUC). However, all 95\% CIs include zero, and effects are imprecisely estimated in this purely illustrative analysis. 

\section{Discussion and Future Work}\label{sec:disc}
Estimating heterogeneous effects of MTPs with SNMMs can be both of scientific interest and important for planning realistic interventions. Suppose policy makers want to know the effect of a campaign to decrease opioid prescription dosing on some continuous quality of life utility measure. They believe the impact of the program might be approximated by the effect of an MTP: `prescribe a 20\% lower dose of opioids than you normally would'. One might be interested in how the effect of this intervention varies with the natural dose. Perhaps it turns out that reductions are more impactful when the natural dose is moderate, as opposed to low or high. Learning this might lead to hypotheses about the mechanism of opioid addiction that could be tested in further studies. The knowledge might also help to design targeted opioid reduction interventions. 

Furthermore, economists who commonly apply DiD can now study MTPs under similar parallel trends assumptions. For example, they can estimate the effect of increasing minimum wage by 1 dollar more than it actually increased, and how this effect varies with the amount it actually increased. This is not an estimand that can be targeted by standard DiD.

There are many directions for future work. First, SNMMs can be estimated under instrumental variable assumptions \citep{robins1994correcting}, and instrumental variable estimation of MTPs would also be of interest. Second, SNMM variants have been developed to estimate effects on survival \citep{picciotto2012structural} and binary \citep{wang2023coherent} outcomes. Effects of MTPs on survival and binary outcomes are of course of interest as well. Third, sensitivity analysis for SNMMs \citep{robins2004optimal,robins2000sensitivity} is well developed and would hopefully easily port over to the MTP setting, but that should be confirmed. Finally, marginal SNMMs that model heterogeneity as a function of a subset of covariates required for adjustment have been developed elsewhere, and it might be useful to transport them to the MTP setting, too.

\section{Acknowledgments}
I am extremely indebted to Jamie Robins for lots of things, but most relevant for this Acknowledgments section is that he suggested the general idea for this paper (`SNMMs for MTPs'). However, all errors are my own.
\bibliography{bibliography}

\begin{thebibliography}{19}
\providecommand{\natexlab}[1]{#1}
\providecommand{\url}[1]{\texttt{#1}}
\expandafter\ifx\csname urlstyle\endcsname\relax
  \providecommand{\doi}[1]{doi: #1}\else
  \providecommand{\doi}{doi: \begingroup \urlstyle{rm}\Url}\fi

\bibitem[Richardson and Robins(2013)]{richardson2013single}
Thomas~S Richardson and James~M Robins.
\newblock Single world intervention graphs (swigs): A unification of the counterfactual and graphical approaches to causality.
\newblock \emph{Center for the Statistics and the Social Sciences, University of Washington Series. Working Paper}, 128\penalty0 (30):\penalty0 2013, 2013.

\bibitem[D{\'\i}az et~al.(2023)D{\'\i}az, Williams, Hoffman, and Schenck]{diaz2023nonparametric}
Iv{\'a}n D{\'\i}az, Nicholas Williams, Katherine~L Hoffman, and Edward~J Schenck.
\newblock Nonparametric causal effects based on longitudinal modified treatment policies.
\newblock \emph{Journal of the American Statistical Association}, 118\penalty0 (542):\penalty0 846--857, 2023.

\bibitem[Robins(1994{\natexlab{a}})]{jamie1994}
James~M. Robins.
\newblock Correcting for non-compliance in randomized trials using structural nested mean models.
\newblock \emph{Communications in Statistics-Theory and Methods}, 23\penalty0 (8), 1994{\natexlab{a}}.

\bibitem[Robins(2004)]{robins2004optimal}
James~M Robins.
\newblock Optimal structural nested models for optimal sequential decisions.
\newblock In \emph{Proceedings of the Second Seattle Symposium in Biostatistics: analysis of correlated data}, pages 189--326. Springer, 2004.

\bibitem[Vansteelandt and Joffe(2014)]{snmm_review2014}
Stijn Vansteelandt and Marshall Joffe.
\newblock Structural nested models and g-estimation: the partially realized promise.
\newblock \emph{Statistical Science}, 29\penalty0 (4):\penalty0 707--731, 2014.

\bibitem[Robins et~al.(2004)Robins, Hern{\'a}n, and SiEBERT]{robins2004effects}
James~M Robins, Miguel~A Hern{\'a}n, and UWE SiEBERT.
\newblock Effects of multiple interventions.
\newblock \emph{Comparative quantification of health risks: global and regional burden of disease attributable to selected major risk factors}, 1:\penalty0 2191--2230, 2004.

\bibitem[Young et~al.(2014)Young, Hern{\'a}n, and Robins]{young2014identification}
Jessica~G Young, Miguel~A Hern{\'a}n, and James~M Robins.
\newblock Identification, estimation and approximation of risk under interventions that depend on the natural value of treatment using observational data.
\newblock \emph{Epidemiologic methods}, 3\penalty0 (1):\penalty0 1--19, 2014.

\bibitem[Haneuse and Rotnitzky(2013)]{haneuse2013estimation}
Sebastian Haneuse and Andrea Rotnitzky.
\newblock Estimation of the effect of interventions that modify the received treatment.
\newblock \emph{Statistics in medicine}, 32\penalty0 (30):\penalty0 5260--5277, 2013.

\bibitem[Mu{\~n}oz and Van Der~Laan(2012)]{munoz2012population}
Iv{\'a}n~D{\'\i}az Mu{\~n}oz and Mark Van Der~Laan.
\newblock Population intervention causal effects based on stochastic interventions.
\newblock \emph{Biometrics}, 68\penalty0 (2):\penalty0 541--549, 2012.

\bibitem[Sani et~al.(2020)Sani, Lee, and Shpitser]{sani2020identification}
Numair Sani, Jaron Lee, and Ilya Shpitser.
\newblock Identification and estimation of causal effects defined by shift interventions.
\newblock In \emph{Conference on Uncertainty in Artificial Intelligence}, pages 949--958. PMLR, 2020.

\bibitem[Shahn et~al.(2022)Shahn, Dukes, Shamsunder, Richardson, Tchetgen, and Robins]{shahn2022structural}
Zach Shahn, Oliver Dukes, Meghana Shamsunder, David Richardson, Eric~Tchetgen Tchetgen, and James Robins.
\newblock Structural nested mean models under parallel trends assumptions.
\newblock \emph{arXiv preprint arXiv:2204.10291}, 2022.

\bibitem[Robins(1986)]{robins1986new}
James Robins.
\newblock A new approach to causal inference in mortality studies with a sustained exposure period—application to control of the healthy worker survivor effect.
\newblock \emph{Mathematical modelling}, 7\penalty0 (9-12):\penalty0 1393--1512, 1986.

\bibitem[Chernozhukov et~al.(2018)Chernozhukov, Chetverikov, Demirer, Duflo, Hansen, Newey, and Robins]{chernozhukov2018double}
Victor Chernozhukov, Denis Chetverikov, Mert Demirer, Esther Duflo, Christian Hansen, Whitney Newey, and James Robins.
\newblock Double/debiased machine learning for treatment and structural parameters, 2018.

\bibitem[{Google LLC}(2022)]{google_mobility_landing_2022}
{Google LLC}.
\newblock Covid-19 community mobility reports.
\newblock \url{https://www.google.com/covid19/mobility/}, 2022.
\newblock Historical data publicly available; updates ended Oct 15, 2022.

\bibitem[{The New York Times}(2025)]{nyt_covid_counties}
{The New York Times}.
\newblock Covid-19 data repository by the new york times.
\newblock \url{https://github.com/nytimes/covid-19-data}, 2025.
\newblock County-level daily cumulative cases/deaths; includes \texttt{us-counties.csv}.

\bibitem[Robins(1994{\natexlab{b}})]{robins1994correcting}
James~M Robins.
\newblock Correcting for non-compliance in randomized trials using structural nested mean models.
\newblock \emph{Communications in Statistics-Theory and methods}, 23\penalty0 (8):\penalty0 2379--2412, 1994{\natexlab{b}}.

\bibitem[Picciotto et~al.(2012)Picciotto, Hern{\'a}n, Page, Young, and Robins]{picciotto2012structural}
Sally Picciotto, Miguel~A Hern{\'a}n, John~H Page, Jessica~G Young, and James~M Robins.
\newblock Structural nested cumulative failure time models to estimate the effects of interventions.
\newblock \emph{Journal of the American Statistical Association}, 107\penalty0 (499):\penalty0 886--900, 2012.

\bibitem[Wang et~al.(2023)Wang, Meng, Richardson, and Robins]{wang2023coherent}
Linbo Wang, Xiang Meng, Thomas~S Richardson, and James~M Robins.
\newblock Coherent modeling of longitudinal causal effects on binary outcomes.
\newblock \emph{Biometrics}, 79\penalty0 (2):\penalty0 775--787, 2023.

\bibitem[Robins et~al.(2000)Robins, Rotnitzky, and Scharfstein]{robins2000sensitivity}
James~M Robins, Andrea Rotnitzky, and Daniel~O Scharfstein.
\newblock Sensitivity analysis for selection bias and unmeasured confounding in missing data and causal inference models.
\newblock In \emph{Statistical models in epidemiology, the environment, and clinical trials}, pages 1--94. Springer, 2000.

\end{thebibliography}
\section*{Appendix}
\subsection{Proof of Theorem 1}\label{appendix:proof1}
\begin{proof}
\textbf{Step 1: Identification.}
By the SNMM restriction,
\[
E\!\left[q\{\mu(A,L)-\mu(g(A,L),L)-\gamma_g(L,A;\psi^\ast)\}\right]=0.
\]
Moreover $E[Y-\mu(A,L)\mid A,L]=0$, hence the augmentation term in
\eqref{eq:point_score} has mean zero. Summing yields $E[\phi(O;\psi^\ast,\eta^\ast)]=0$.\\
\\
\textbf{Step 2: Orthogonality w.r.t.\ $\mu$.}
Consider $\mu_t=\mu^{*}+t\,\delta\mu+o(t)$ with $\pi$ fixed.
Only $Y-\mu_t$ and the identifying bracket depend on $t$.
Differentiate the expectation of \eqref{eq:point_score} at $t=0$:
\[
\left.\frac{d}{dt}\right|_{0}E\!\left[(q-\tilde q)(Y-\mu_t)\right]
= -\,E\!\left[(q-\tilde q)\delta\mu\right]
= -\,\langle q-\tilde q,\delta\mu\rangle_w.
\]
For the identifying term,
\[
\left.\frac{d}{dt}\right|_{0}
E\!\left[q\{\mu_t - T_g\mu_t - \gamma_g(\psi^{*})\}\right]
= E\!\left[q\,\delta\mu\right]-E\!\left[q\,T_g\delta\mu\right].
\]
Rewrite under $\langle\cdot,\cdot\rangle_w$ and apply the definition of $\tilde{q}$:
\[
E[q\,\delta\mu] - E[q\,T_g\delta\mu]
= \langle q,\delta\mu\rangle_w - \langle q, T_g\delta\mu\rangle_w
= \langle q,\delta\mu\rangle_w - \langle \tilde q,\delta\mu\rangle_w
= \langle q-\tilde q,\delta\mu\rangle_w.
\]
Summing the two derivatives: $-\langle q-\tilde q,\delta\mu\rangle_w+\langle q-\tilde q,\delta\mu\rangle_w=0$.\\
\\
\textbf{Step 3: Orthogonality w.r.t.\ $\pi$.}
Let $\pi_t=\pi^{*}+t\,\delta\pi+o(t)$ and $w_t=\pi_t$.
Then both $w_t$ and $\tilde q_t:=T_g^{\dagger}q$ vary with $t$.
Write
\[
F(t):=E\!\left[\phi(O;\psi^{*},\mu^{*},\pi_t)\right]
= \underbrace{E\!\left[(q-\tilde q_t)(Y-\mu^{*})\,\right]}_{=:A(t)}
+ \underbrace{E\!\left[q\{\mu^{*}-T_g\mu^{*}-\gamma_g(\psi^{*})\}\right]}_{=:B\ \text{(const.\ in $t$)}}.
\]
Hence $F'(0)=A'(0)$.
Differentiate $A(t)$ at $0$:
\[
A'(0)
=
-\,E\!\big[\dot{\tilde q}\,(Y-\mu^{*})\,w^{*}\big]
\;+\;
E\!\big[(q-\tilde q^{*})(Y-\mu^{*})\,\dot w\big],
\]
where dots denote $t$-derivatives at $0$, and $w^{*}=\pi^{*}$.
Now \emph{condition on $(A,L)$}. Since $E[Y-\mu^{*}(A,L)\mid A,L]=0$,
both terms vanish:
\[
E\!\big[\dot{\tilde q}\,(Y-\mu^{*})\,w^{*}\big]
=E\!\big[\,E[Y-\mu^{*}\mid A,L]\;\dot{\tilde q}\,w^{*}\,\big]=0,
\]
and similarly for the second term. Therefore $A'(0)=0$, i.e.\ $F'(0)=0$.

An equivalent route is to differentiate the identity
$\langle q,T_g h\rangle_{w_t}=\langle \tilde q_t,h\rangle_{w_t}$ at $t=0$
with $h(a,l)=Y-\mu^{*}(a,l)$. This yields
\[
E\!\left[(q-\tilde q^{*})\,T_g h\,\dot w\right]=E\!\left[\dot{\tilde q}\,h\,w^{*}\right],
\]
which is exactly the cancellation needed in $A'(0)$ when $h$ is replaced by
$Y-\mu^{*}$.
\end{proof}
\subsection{Proof of Theorem 2}\label{appendix:proof2}

\begin{proof}
\textbf{Step 1: Identification.}
By construction, for each $t$,
\(
E\!\big[V_{t+1}-\mu_t(H_t,A_t)\mid H_t,A_t\big]=0,
\)
hence the augmentation term of $\phi_t$ has mean zero.
The longitudinal SNMM restriction gives
\(
E\!\big[q_t\{\mu_t(H_t,A_t)-\mu_t(H_t,g_t(A_t,H_t))-\gamma_t(H_t,A_t;\psi^{*})\}\big]=0.
\)
Summing over $t$ yields $E[\Phi(O;\psi^{*},\eta^{*})]=0$.\\
\\
\textbf{Step 2: Orthogonality w.r.t.\ $\mu_t$.}
Fix $t$ and perturb $\mu_{t,\varepsilon}=\mu_t^{*}+\varepsilon\delta\mu_t+o(\varepsilon)$,
holding all $\{\mu_s,\pi_s\}_{s\neq t}$ and $\pi_t$ fixed. Note that
$V_{t+1}$ is treated as fixed when differentiating $\phi_t$ in its own nuisances.
Differentiate the expectation of $\phi_t$ at $\varepsilon=0$:
\[
\left.\frac{d}{d\varepsilon}\right|_{0}
E\!\left[(q_t-\tilde q_t)(V_{t+1}-\mu_{t,\varepsilon}(H_t,A_t))\right]
= -\,\langle q_t-\tilde q_t,\delta\mu_t\rangle_{w_t},
\]
and
\[
\left.\frac{d}{d\varepsilon}\right|_{0}
E\!\left[q_t\{\mu_{t,\varepsilon} - T_{g,t}\mu_{t,\varepsilon} - \gamma_t(\psi^{*})\}\right]
= \langle q_t,\delta\mu_t\rangle_{w_t} - \langle q_t, T_{g,t}\delta\mu_t\rangle_{w_t}
= \langle q_t-\tilde q_t,\delta\mu_t\rangle_{w_t},
\]
by the adjoint identity for time $t$. The two derivatives cancel exactly.
Because $\phi_t$ does not depend on $\mu_s$ for $s\neq t$ (given $V_{t+1}$ fixed),
we obtain $D_{\mu_t}E[\phi_t]=0$.

\textbf{Step 3: Orthogonality w.r.t. $\pi_t$.}
Perturb $\pi_{t,\varepsilon}$; then $w_{t,\varepsilon}=\pi_{t,\varepsilon}$ and
$\tilde q_{t,\varepsilon}=T_{g,t}^{\dagger,w_{t,\varepsilon}}q_t$ vary with $\varepsilon$.
Write
\[
F_t(\varepsilon):=E[\phi_t(O;\psi^{*},\mu_t^{*},\pi_{t,\varepsilon})]
= \underbrace{E\!\left[(q_t-\tilde q_{t,\varepsilon})\,(V_{t+1}-\mu_t^{*})\,\right]}_{=:A_t(\varepsilon)}
+ \underbrace{E\!\left[q_t\{\mu_t^{*}-T_{g,t}\mu_t^{*}-\gamma_t(\psi^{*})\}\right]}_{=:B_t\ \text{const.}}.
\]
Hence $F_t'(0)=A_t'(0)$.
Differentiate at $0$:
\[
A_t'(0)
=
-\,E\!\big[\dot{\tilde q}_t\,(V_{t+1}-\mu_t^{*})\,w_t^{*}\big]
\;+\;
E\!\big[(q_t-\tilde q_t^{*})\,(V_{t+1}-\mu_t^{*})\,\dot w_t\big].
\]
\emph{Conditioning on $(H_t,A_t)$} and using
$E[V_{t+1}-\mu_t^{*}(H_t,A_t)\mid H_t,A_t]=0$, both terms are zero,
hence $A_t'(0)=0$.
Equivalently, differentiate the time-$t$ adjoint identity
$\langle q_t, T_{g,t}h\rangle_{w_{t,\varepsilon}}=
\langle \tilde q_{t,\varepsilon}, h\rangle_{w_{t,\varepsilon}}$
at $\varepsilon=0$ with $h(H_t,A_t)=V_{t+1}-\mu_t^{*}(H_t,A_t)$ to see the same cancellation.

Summing the time-local derivatives over $t$ gives orthogonality for $\Phi$.
\end{proof}

\end{document}